\documentclass[aps,prl,twocolumn,showpacs,groupedaddress]{revtex4-1}

\usepackage{graphics,graphicx,psfrag,color,float}
\usepackage{epsfig,bbold,bbm}
\usepackage{amssymb}
\usepackage{amsthm}
\usepackage[colorlinks=true,citecolor=blue,urlcolor=blue,linkcolor=blue]{hyperref}
\usepackage[fleqn]{amsmath}
\usepackage{nccmath}

\usepackage{xspace}
\usepackage{stmaryrd}
\usepackage{braket}
\usepackage{xfrac}

\newtheorem{theorem}{Theorem}
\newtheorem{lemma}{Lemma}

\newcommand{\beq}{\begin{equation}}
\newcommand{\enq}{\end{equation}}
\newcommand{\bel}{\begin{lemma}}
\newcommand{\enl}{\end{lemma}}
\newcommand{\bet}{\begin{theorem}}
\newcommand{\ent}{\end{theorem}}

\newcommand{\tr}{\mathrm{Tr}}

\newcommand*{\cH}{\mathcal{H}}
\newcommand*{\cF}{\mathcal{F}}

\newcommand*{\cD}{\mathcal{D}}
\newcommand*{\cG}{\mathcal{G}}
\newcommand*{\cK}{\mathcal{K}}
\newcommand*{\cN}{\mathcal{N}}
\newcommand*{\cS}{\mathcal{S}}

\newcommand*{\cT}{\mathcal{T}}

\newcommand*{\cE}{\mathcal{E}}

\newcommand*{\cR}{\mathcal{R}}

\newcommand*{\renyi}{R\'{e}nyi }

\begin{document}


\title{Fundamental bound on the reliability of quantum information transmission}



\author{Naresh Sharma}
\email{nsharma@tifr.res.in}



\author{Naqueeb Ahmad Warsi}
\email{naqueeb@tifr.res.in}

\affiliation{
Tata Institute of Fundamental Research (TIFR),
Mumbai 400005, India}


\begin{abstract}
Information theory tells us that if the rate of sending information across a noisy channel were
above the capacity of that channel, then the transmission would necessarily
be unreliable. For classical
information sent over classical or quantum channels, one could, under certain conditions,
make a stronger statement that the reliability of the transmission shall decay exponentially
to zero with the number of channel uses and the proof of this statement typically
relies on a certain fundamental bound on the reliability of the transmission. Such a statement
or the bound has never been
given for sending quantum information. We give this bound
and then use it to give the first 
example where the reliability of sending quantum information at rates above the capacity
decays exponentially to zero. We also show that our framework can be used
for proving generalized bounds on the reliability.
\end{abstract}

\maketitle
Capacity of a given channel is defined as the highest rate of sending information (measured
as the amount of information sent per channel use) reliably in the limit of large
number of channel uses \cite{shannon1948, covertom, wilde-book}. Converse of
the channel capacity theorem
tells us that sending information at rates higher than capacity would necessarily
be unreliable. A strong converse additionally tells us that the reliability would be
very small and, in some cases more explicitly, would decay
exponentially to zero with the number of channel uses. Not all channels have a strong
converse \cite{no-strong-converse-2011}.

Such strong converses are available for sending classical information across classical
or quantum channels (under certain conditions) and are typically shown using a fundamental
bound on the reliability. But, somewhat surprisingly, there has been
no such strong converse when quantum information is sent across a quantum channel
and an equivalent bound has been unknown.
We first prove this bound in full generality and then apply it to give the first
example of a strong converse for quantum information transfer where the reliability
decays exponentially to zero with the number of channel uses.

Strong converse establishes capacity as a sharp thresh- old for information transmission
and is clearly of great theoretical interest. It also has interesting applications in cryptography.
Let Alice have an unlimited noise-free quantum memory to store qubits while Bob has a
noisy quantum memory (also called the noisy-storage assumption). If the strong converse
holds for the quantum channel modelling the noise that acts on Bob's memory, then Alice
and Bob can implement any two-party cryptographic task securely
\cite{crypto-sconv-2011}.

We now provide a more detailed but high level overview of our results.
A protocol to transfer information (classical or quantum) across a noisy
communication channel is characterised by the amount of information ($\cR$) it conveys
and the reliability ($\mathbbm{F}$) it promises.
Typical definitions of reliability ensure that $\mathbbm{F} \in [0,1]$, where
$\mathbbm{F} \approx 1$ would imply a highly reliable information transfer,
i.e., information sent and reconstructed at the receiver
are very close to each other ($\mathbbm{F}=1$ implies
an exact match) and $\mathbbm{F} \approx 0$ would imply a highly unreliable transmission.

Information could be classical or quantum.
A classical information is an unknown sequence of bits
(such as an email message) that Alice wants to send to Bob.
A quantum information transfer can also be looked upon as entanglement transfer
\cite{wilde-book}.
Alice has a quantum system $S$ (information) that is entangled with a reference system $A$
and Alice (who doesn't have access to $A$) wishes to send a quantum system
through a noisy environment (that doesn't act upon $A$)
such that at the end of the protocol,
the state of $A$ and Bob's system (say $\hat{S}$) is close to the state of $A$ and $S$.

Fundamental bound that we seek for all $s \in [-\beta,0)$
and protocol parameters $\pmb{\alpha}$ is given by
\beq
\label{fund-bound}
\mathbbm{F} \leq e^{s \cR - E_0(s,\pmb{\alpha})},
\enq
where $E_0(0,\pmb{\alpha}) = 0$, the derivative of $E_0(s,\pmb{\alpha})$ w.r.t. $s$
at $s=0$ gives us a measure of information that could be transferred across the
channel reliably, and $\beta$ is a constant independent of $\pmb{\alpha}$
and $\cR$ that, for our purposes, is $0.5$.

$E_0(s,\pmb{\alpha}) - s\cR$ is known as the Gallager's exponent named after
R. G. Gallager who first proposed it in a different setting \cite{gallager-expo-1965}. The
bound in Eq. \eqref{fund-bound}
was shown when classical information is sent across a classical 
channel (Arimoto \cite{arimoto-1973-converse}) and quantum channel (Ogawa and
Nagaoka \cite{ogawa-1999-converse}). Winter gave another proof of the strong converse
for sending classical information over quantum channels without the Gallager's exponent
\cite{winter-99-converse}. Extensions of the above results
are due to K\"{o}nig and Wehner (Ref. \cite{konig-2009-converse}) and further upper bounds
to fidelity for entanglement unassisted and assisted codes are given by Matthews and
Wehner \cite{matt-wehner}.

The search for quantum Gallager's exponent
when quantum information is sent across a quantum channel has been a longstanding
problem and we provide it in this paper. Table \ref{tab1} lists these various cases.

Our proof relies on using the monotonicity property (mentioned below) satisfied by
many information divergences. The idea of proving bounds on the reliability for classical
protocols using monotonicity dates back to Blahut's work \cite{blahut1976} and has been
used further more recently \cite{han-verdu1994, sharmapra2008, polyanskiy-2010-converse}.

We now provide a brief and heuristic explanation as to why this bound is considered 
fundamental.
Let us define for a single use of channel that
$I(\pmb{\alpha}) = \partial E_0(s,\pmb{\alpha})/\partial s |_{s=0}$
and $C = \max_{\pmb{\alpha}^\prime} I(\pmb{\alpha})$ be called the channel capacity,
where $\pmb{\alpha}^\prime$ is the part of $\pmb{\alpha}$ that can be changed by
fine-tuning the protocol \cite{shannon1948, covertom, wilde-book}.
There are parameters
in the setup that can't be changed such as the channel and there may be some practical
constraints such as energy used for transmission that the protocol must obey.
Since $E_0$ obeys $E_0(0,\pmb{\alpha})=0$,
for a negative $s$ near $0$, $-E_0(s,\pmb{\alpha})$
$\approx -s I(\pmb{\alpha})$ $ \leq -s C$ and the above bound could be weakened to give
$\mathbbm{F} \lessapprox e^{s (\cR - C)}$. Hence, if $\cR > C$, then $\mathbbm{F}$ is always exponentially
bounded away from $1$. If we use the channel $n$ times for sending
$n \cR$ amount of information,
then we could, under certain conditions, write the above bound as
$\mathbbm{F} \lessapprox e^{s n (\cR - C)}$. If $\cR > C$, then $\mathbbm{F} \to 0$ exponentially with $n$, i.e.,
if we are pumping information into the channel higher
than the capacity, then the transmission would be quite unreliable.

\begin{table}[t]
\caption{Gallager's exponent (that gives an exponential upper
bound on reliability) for various cases.}
\begin{center}
\begin{ruledtabular}
\begin{tabular}{c|c|c}
{\bf Information } & {\bf Channel } & {\bf Proposed by} \\ \hline
Classical & Classical & Arimoto (1973) \\ \hline
Classical & Quantum & Ogawa \& Nagaoka (1999) \\ \hline
Quantum & Quantum & (this paper)
\end{tabular}
\end{ruledtabular}
\end{center}
\label{tab1}
\end{table}

We shall frequently deal with the quantum \renyi divergences in this paper that for parameter
$\lambda \geq 0$ are given by
\beq
D_\lambda(\rho || \sigma) = \frac{1}{\lambda - 1} \ln \tr \rho^{\lambda} \sigma^{1-\lambda},
\enq
where limit is taken at $\lambda = 1$. We shall confine ourselves with $\lambda \in (1,2]$
in this paper and deal with finite dimensional quantum systems.
The following two properties are needed later. \\
\noindent {\bf Property 1}: It has been shown (see Example 4.5 in Ref. \cite{hiai-2011})
that for the chosen range of $\lambda$, $D_\lambda$
satisfies the monotonicity property, i.e., for any two un-normalised density matrices (that
are positive but need not have a unit trace) $\rho$, $\sigma$ and a completely
positive and trace preserving (CPTP) quantum operation
$\cN$ acting on them, we have
\beq
\label{mono}
D_\lambda(\rho || \sigma) \geq D_\lambda\left[\cN(\rho) || \cN(\sigma)\right].
\enq
\noindent {\bf Property 2}: We shall also need the following queer property that
is not difficult to prove. Let $\Pi_0 = \ket{0}\bra{0}$ and $\Pi_1 = \ket{1}\bra{1}$
be two projectors with $\Pi_0 + \Pi_1 = \mathbbm{1}$.
Let $\alpha \in [0,1]$, $\beta \in (0,1]$, $\rho = \alpha \Pi_0 + (1-\alpha) \Pi_1$,
$\sigma = \beta \Pi_0 + (1/\beta-\beta) \Pi_1$, and let us define
\beq
\label{dummy1}
{\mathbbm{D}}_\lambda(\alpha || \beta) := D_{\lambda}(\rho || \sigma).
\enq
Note that $\sigma \geq 0$ but does not have unit trace.
Then ${\mathbbm{D}}_\lambda(\alpha||\beta)$ is independent of the choice of
$\{\Pi_0,\Pi_1\}$ and increasing for all $\alpha \geq \beta$.

We now derive a quantity from the \renyi divergence as
\begin{align}
K_{\lambda}(A \rangle B)_\rho & := \inf_{\sigma^B \in \cS(\cH_B)} D_{\lambda}(\rho^{AB} ||
\mathbbm{1} \otimes \sigma^B),
\end{align}
where $\cH_B$ is the Hilbert space describing quantum
system $B$ and
$\cS(\cH_B)$ is the set of all density matrices of $\cH_B$, and $\mathbbm{1}$ is the
identity matrix whose dimensions should be clear from the context.
Csisz\'{a}r defined a similar quantity in the classical case and related it
to the Gallager's exponent \cite{csiszar-1995-rate}.
The following properties of $K_{\lambda}(A \rangle B)_\rho$ would be useful later.
\begin{lemma}
Let $\cE^{B \to C}$ be a quantum operation and
$\rho^{AC} = \cE^{B \to C}(\rho^{AB})$. Then
\[
K_{\lambda}(A \rangle B)_\rho \geq K_{\lambda}(A \rangle C)_\rho.
\]
\end{lemma}
\begin{proof}
See Appendix.
\end{proof}
\begin{lemma}
\label{exp1}
Let $\rho^{A A^\prime}$ be any quantum state in $AA^\prime$, and
$\rho^{A B} = \cN^{A^\prime \to B} (\rho^{AA^\prime})$. Then
\[
K_{\lambda}(A \rangle B)_\rho = \frac{\lambda}{1 - \lambda}
E_0(\lambda^{-1}-1, \cN^{A^\prime \to B})_\rho,
\]
where for $s = \lambda^{-1} - 1$,
\[
E_0(s, \cN^{A^\prime \to B})_\rho :=
-\ln \tr \left[ \tr_{A} \left( \rho^{A B} \right)^{\frac{1}{s+1}}
\right]^{s+1}. ~~
\]
\end{lemma}
\begin{proof}
See Appendix.
\end{proof}

{\bf Information processing task}:
Suppose a quantum system $S$ and a reference system $A$ have a state $\ket{\phi}^{AS}$.
Alice only has access to the system $S$ and not to $A$. Alice wants to send her part of the
shared state with $A$ to Bob using $n$ independent uses of a quantum
channel $\cN^{A^{\prime} \to B}$ such that at the end of the communication protocol chain,
Bob's shared state with the reference $A$ is arbitrarily close to the state Alice shared
with $A$. We shall call $\cR$ to be the communication rate and is given by
$\cR := \ln |S| / n$, where $|S|$ is the dimension of $\cH_S$.
We shall assume that the state of $S$ is given by
${\mathbbm{1}}/|S|$, i.e., the completely mixed state.

To this end, Alice performs an encoding operation given by $\cE^{S \to A^{\prime n}}$ to get
$\rho^{AA^{\prime n}} = \cE^{S \to A^{\prime n}} \left( \phi^{AS} \right)$.
Alice transmits the system $A^{\prime n}$ over
$\cN^{A^{\prime n} \to B^n} = \left(\cN^{A^\prime \to B} \right)^{\otimes n}$
and Bob receives the state
$\rho^{AB^n} = \cN^{A^{\prime n} \to B^n} \left[ \cE^{S \to A^{\prime n}}
\left( \phi^{AS} \right) \right]$.
Bob applies a decoding operation on its part
of the received state to get
$\rho^{A\hat{S}} = \cT^{B^n \to \hat{S}} \left\{
\cN^{A^{\prime n} \to B^n} \left[ \cE^{S \to A^{\prime n}} \left( \phi^{AS} \right) \right]
\right\}$.
The performance of the protocol is quantified by the fidelity given by
$F(\phi^{AS}, \rho^{A\hat{S}}) = \bra{\phi}^{AS} \rho^{A\hat{S}} \ket{\phi}^{AS}$.
If a protocol promises a fidelity not smaller than $\mathbb{F}$,
then we shall refer to such a protocol as a $(n,\cR,1-\mathbb{F})$ code.

The maximum rate per channel use for this protocol in the limit of large number of
channel uses and
fidelity arbitrarily close to $1$ was proved in a series of papers (see Refs.
\cite{schumacher1996,
schumacher-nielsen-1996, barnum-nielsen-1998, barnum-knill-2000, lloyd-1997,
shor-cap-2002, devetak-2005, hayden-shor-2008}).
Let the coherent information
of the channel $\cN^{A^\prime \to B}$ be defined as
$Q(\cN) := \max_{\rho^{AA^\prime}} I(A \rangle B)_\sigma$,
where $\sigma^{A B} = \cN^{A^\prime \to B} (\rho^{AA^\prime})$,
$I(A \rangle B)_\sigma := H(B)_\sigma - H(A,B)_\sigma$, and
$H(A)_\sigma$ is the von Neumann entropy of a quantum state $\sigma$ in system $A$
given by $H(A)_\sigma = -\tr \sigma \ln \sigma$.
The capacity of the channel is now given by the regularisation
$Q_{\mathrm{reg}}(\cN) := \lim_{n \to \infty} Q(\cN^{\otimes n})/n$.

We now prove an inequality involving the fidelity and the rate.
\bet
\label{quant-conv}
For $\mathbb{F} \geq e^{-n\cR}$, any $(n,\cR,1-\mathbb{F})$ code satisfies
\[
{\mathbbm{D}}_\lambda(\mathbb{F} || e^{-n\cR}) \leq K_{\lambda}(A \rangle B^n)_\rho.
\]
\ent
\begin{proof}
Let $\{\ket{i}^{AS}\}$ be an orthonormal basis for $\cH_{AS}$ with
$\ket{1}^{AS} = \ket{\phi}^{AS}$. Consider a CPTP quantum map
$\cF^{A\hat{S} \to C}$ where $|C|=2$ with Kraus operators
$\ket{0}^C \bra{1}^{AS}$, and $\{\ket{1}^C \bra{i}^{AS} \}$, $i=2,3,...,|AS|$. Let
$\Pi_0^C = 0^C$ and $\Pi_1^C = 1^C$. Then for all $\sigma^{\hat{S}}$, we have
$\cF(\rho^{A\hat{S}}) = \mathbb{F}^\prime \Pi^C_0 + (1-\mathbb{F}^\prime) \Pi^C_1$,
$\cF({\mathbbm{1}} \otimes \sigma^{\hat{S}}) = e^{-n\cR} \Pi^C_0 +
(e^{n\cR}-e^{-n\cR}) \Pi^C_1$,
where $\mathbb{F}^\prime = \bra{\phi}^{AS} \rho^{A\hat{S}} \ket{\phi}^{AS}$.
We now have the following inequalities
\begin{align*}
K_{\lambda}(A \rangle B^n)_\rho & \stackrel{a}{\geq} \
\inf_{\sigma^{\hat{S}}} D_{\lambda}(\rho^{A\hat{S}} || {\mathbbm{1}} \otimes \sigma^{\hat{S}}) \\
& \stackrel{b}{\geq} \inf_{\sigma^{\hat{S}}}
D_{\lambda} \Big[ \mathbb{F}^\prime \Pi^C_0 + (1-\mathbb{F}^\prime) \Pi^C_1 || ~~~~~~~ \nonumber \\
& ~~~~~~ e^{-n\cR} \Pi^C_0 + (e^{n\cR}-e^{-n\cR}) \Pi^C_1 \Big] \\
& \stackrel{c}{=} {\mathbbm{D}}_\lambda(\mathbb{F}^\prime || e^{-n\cR}) \\
& \stackrel{d}{\geq} {\mathbbm{D}}_\lambda(\mathbb{F} || e^{-n\cR}),
\end{align*}
where $a$ and $b$ follow from the data processing inequality and the definition
of $K_{\lambda}$, $c$ follows since the quantity
${\mathbbm{D}}_\lambda(\mathbb{F}^\prime || e^{-n\cR})$ is independent of
$\sigma^{\hat{S}}$, and
$d$ from the Property 2 of ${\mathbbm{D}}_\lambda$.
\end{proof}
The constraint  $\mathbb{F} \geq e^{-n\cR}$ may not be
seen as weakening the bound because, if
the constraint is violated, i.e., $\mathbb{F} \leq e^{-n\cR}$, then this, by itself,
would imply an exponential convergence of $\mathbb{F}$ to $0$.
We first note that
\beq
\label{temp3002}
{\mathbbm{D}}_\lambda(\mathbb{F} || e^{-n\cR}) \geq \frac{\lambda}{\lambda - 1}
\ln \mathbb{F} + n \cR
\enq
and it follows from Lemma \ref{exp1} and Theorem \ref{quant-conv} that
\beq
\mathbbm{F} \leq e^{s n \cR - E_0[s, (\cN^{A^\prime \to B})^{\otimes n}]_\rho},
\enq
which gives us the quantum Gallager's exponent.
The properties of $E_0$ are studied by the following theorem.
\begin{theorem}
\label{qexp}
For any quantum state $\sigma^{AB}$, $s \in [-1/2,0)$, the function
\[
g(s) := -\ln \tr \left[ \tr_A \left( \sigma^{AB} \right)^{1/(s+1)} \right]^{s+1},
\]
satisfies
\begin{align*}
g(0) & = 0, \\
\label{coherent}
\frac{\partial g(s)}{\partial s}\Big|_{s=0} & = I(A\rangle B)_{\sigma},
\end{align*}
and $g(s) + (s+1) \ln |A|$ is an increasing function in $s$.
\end{theorem}
\begin{proof}
See Appendix.
\end{proof}

We note here that only the two above mentioned properties of the quantum \renyi
divergence are used for our results. Hence,
if the \renyi divergence is replaced by any other divergence that
satisfies these two properties, then Theorem \ref{quant-conv} shall hold for
that divergence as well. The non-commutative hockey-stick divergence
that we now define is one such example that for $\rho, \sigma \geq 0$, and $\gamma \geq 1$ is given by $\cD(\rho || \sigma) = \tr(\rho - \gamma \sigma)^+$, where $\kappa^+$ is the
positive part of a Hermitian matrix
$\kappa = \kappa^+ - \kappa^-$, $\kappa^+, \kappa^- \geq 0$. It can be regarded as
a non-commutative generalisation of the classical $f$-relative entropy (see
Ref. \cite{csiszar-f-1970})
using the hockey stick function $f(x) = (x-\gamma)^+$ \cite{hull-finance-book}.
We similarly define a derived quantity as
\begin{align*}
\cK(A \rangle B)_\rho & := \inf_{\sigma^B \in \cS(\cH_B)} \cD(\rho^{AB} || \mathbbm{1} \otimes \sigma^B).
\end{align*}
%

\section{Quantum erasure channel with maximally entangled inputs}

We show that the fidelity would decrease exponentially with the number of channel uses for
rates above capacity for maximally entangled inputs that have the full Schmidt rank.

A quantum erasure channel transmits the input state with probability $1-p$ and
``erases" it, i.e., replaces it with an orthogonal erasure state with probability $p$
\cite{grassl-1997} (see also Ref. \cite{bennett-1996-pra}).
The dimension of the output Hilbert space is one larger than that of the input.

A quantum erasure channel $\cN_{p}^{A^\prime \to B}$, defined in
Ref. \cite{wilde-book}, is given by the following Kraus operators
$\Big\{ \sqrt{(1-p)} \sum_{i=1}^{|A^\prime|} \ket{i}^B\bra{i}^{A^\prime},$
$\sqrt{p}\ket {e}^B\bra{1}^{A^\prime},...,\sqrt{p}\ket {e}^B\bra{|A^\prime|}^{A^\prime} \Big\}$, 
$i = 1,..., |A^\prime|$, $p \in [0,1]$, $|B| = |A^\prime|+1$, 
$\left\{\ket{i}^{A^\prime}\right\}, \left\{\ket{i}^B\right\}$ are orthonormal bases in
$\cH_{A^\prime}$ and $\cH_B$ respectively, and $\ket{e}^B = \ket{j}^B$ for
$j = |B|$. The action of the channel  can be understood as follows 
\[
\cN_{p}^{{A^\prime} \to B}(\rho^{A{A^\prime}}) = (1-p) \sigma^{AB} + p \rho^A \otimes
\ket{e} \bra{e}^B.
\]
Let $\sigma^{AB}$ $= \cG^{A^\prime \to B} (\rho^{AA^\prime})$, where $\cG$
increases the dimension but leaves the state intact.
Then with probability $1-p$, the channel 
leaves the state as $\sigma^{AB}$
and with probability $p$, it erases the state and replaces by $\ket{e}^B$.
It is not difficult to see that $\sigma^{AB}$ is orthogonal
to $\rho^A \otimes \ket{e} \bra{e}^B$.

Taking this further for $n$ channel uses, let $\sigma^{AB^n}$
$= ( \cG^{A^\prime \to B} )^{\otimes n} (\rho^{AA^{\prime n}})$.
The output can be written as the sum of $2^n$ orthogonal density
matrices where each of these matrices results from $i$ erasures $i \in\left\{0,...,n \right\}$
and this occurs with probability $(1-p)^{n-i}p^i$. The number of states that have
suffered exactly $i$ erasures is ${n \choose i}$.

Let  $B_{i_1} \cdots B_{i_{n-k}}$ be the
quantum systems that have not suffered erasures and we could write the state in this
case using $\sigma^{AB^n}$ as
\begin{fleqn}
\begin{align*}
\zeta_{i_1,...,i_{n-k}}^{AB_{i_1} \cdots B_{i_n}}
~ = ~ & \sigma^{AB_{i_1} \cdots B_{i_{n-k}}} \otimes
\bigotimes_{j=1}^k \ket{e} \bra{e}^{B_{i_{n-k+j}}}.
\end{align*}
\end{fleqn}
It now follows that
\beq
\label{outstate}
\rho^{A B^n} = \sum_{2^n {\mathrm{terms}}}
\alpha_{k,n}  \times \zeta_{i_1,...,i_{n-k}}^{AB_{i_1} \cdots B_{i_n}},
\enq
where $\alpha_{k,n} = (1-p)^{n-k} p^k$.

To prove the strong converse, we find an upper bound for $K_{\lambda}(A \rangle B^n)$.
We assume that $\rho^{A A^{\prime n}}$ is a maximally entangled state with a Schmidt
rank of $d_A^n$ where $d_A = |A^\prime|$. Note that this is the capacity-achieving input
for this channel and $Q(\cN) = (1-2p)^+ \ln d_A$ is the single-letter
quantum capacity for this channel \cite{bennett-1997} (see also Ref. \cite{wilde-book}).
Note that $d_A^k \times \rho^{A A^\prime_1 \cdots A^\prime_{n-k}}$ is a projector of
rank $d_A^k$ and $\rho^{A^\prime_1 \cdots A^\prime_{n-k}}$ is
the maximally mixed state.

\begin{theorem}
\label{erasureconverse}
The strong converse holds for the quantum erasure channel for the above
chosen maximally entangled channel inputs.
\end{theorem}
\begin{proof}
Note the following set of inequalities for $s = \lambda^{-1}-1$, $\lambda \in (1,2]$
\begin{fleqn}
\begin{align}
K_{\lambda}(A \rangle B^n) &\stackrel{a}= -\frac{1}{s}
\ln \tr\left[\tr_{A}(\rho^{AB^n})^\lambda\right]^{\frac{1}{\lambda}} \nonumber \\
& \stackrel{b} =  -\frac{1}{s} \ln \sum_{2^n {\mathrm{terms}}}
\alpha_{k,n}\tr\left[\tr_A(\zeta_{i_1,...,i_{n-k}}^{AB_{i_1} \cdots B_{i_n}})^\lambda\right]^\frac{1}{\lambda} ~~~~~~ \nonumber \\
\label{last}
& \stackrel{c} \leq -\frac{1}{s} \ln \Bigg\{ \sum_{2^n {\mathrm{terms}}} 
\alpha_{k,n} \nonumber \\
& ~~~~~~~~~~ \exp \left[\frac{-K_{\lambda}
\left(A \rangle A^\prime_{i_1} \cdots A^\prime_{i_{n-k}}\right)}{s} \right] \Bigg\} \nonumber,
\end{align}
\end{fleqn}
where $a$ follows from Lemma \ref{qsibson}, $b$ follows from \eqref{outstate} and the orthogonality of $\zeta$'s
and $c$ follows because $K_{\lambda}$ satisfies monotonicity
and Lemma \ref{qsibson}. Using the fact that 
$d_A^k\times\rho^{AA^\prime_{i_1} \cdots A^\prime_{i_{n-k}}}$ 
is a projector of rank $d_A^k$, we get 
$K_{\lambda}(A \rangle B^n) \leq n E_0(s)/s$ where we define (with some abuse of notation)
\[
E_0(s):= -\ln \Big[(1-p) d_A^{-s} + p d_A^{s} \Big]
\]
and $E_0(0)=0$.
Using \eqref{temp3002}, we have
\[
\mathbb{F} \leq \exp \left\{ n \left[ s\cR - E_0(s) \right] \right\}.
\]
Furthermore, for $p \in [0,1/2]$,
\[
\lim_{s \uparrow 0} \frac{E_0(s)}{s} = Q(\cN).
\]
Hence, for all $\cR > Q(\cN)$, $\exists$ $s \in [-1/2,0)$ s.t.
$\cR - E_0(s)/s$ $> 0$, and thus the strong converse holds.
For $p > 1/2$, $E_0^\prime(0)$ $< 0$ and hence, using similar arguments
as above, for any $\cR>0$, the strong converse holds.
\end{proof}

An alternate proof of Theorem \ref{erasureconverse} using the hockey stick divergence
is provided in the Appendix.

To summarise our results, we have given an exponential upper bound on the
reliability of quantum information transmission. The bound is fundamental in
the same vein as the bounds known for transmission of classical information
across classical/quantum channels (see Refs.
\cite{arimoto-1973-converse, ogawa-1999-converse, konig-2009-converse})
and holds under general conditions. We then apply
our bound to yield the first known example for exponential decay of reliability at
rates above capacity for quantum information transmission.

The authors gratefully acknowledge the comments by A. Winter.

\bibliography{master}

\begin{thebibliography}{31}%
\makeatletter
\providecommand \@ifxundefined [1]{%
 \@ifx{#1\undefined}
}%
\providecommand \@ifnum [1]{%
 \ifnum #1\expandafter \@firstoftwo
 \else \expandafter \@secondoftwo
 \fi
}%
\providecommand \@ifx [1]{%
 \ifx #1\expandafter \@firstoftwo
 \else \expandafter \@secondoftwo
 \fi
}%
\providecommand \natexlab [1]{#1}%
\providecommand \enquote  [1]{``#1''}%
\providecommand \bibnamefont  [1]{#1}%
\providecommand \bibfnamefont [1]{#1}%
\providecommand \citenamefont [1]{#1}%
\providecommand \href@noop [0]{\@secondoftwo}%
\providecommand \href [0]{\begingroup \@sanitize@url \@href}%
\providecommand \@href[1]{\@@startlink{#1}\@@href}%
\providecommand \@@href[1]{\endgroup#1\@@endlink}%
\providecommand \@sanitize@url [0]{\catcode `\\12\catcode `\$12\catcode
  `\&12\catcode `\#12\catcode `\^12\catcode `\_12\catcode `\%12\relax}%
\providecommand \@@startlink[1]{}%
\providecommand \@@endlink[0]{}%
\providecommand \url  [0]{\begingroup\@sanitize@url \@url }%
\providecommand \@url [1]{\endgroup\@href {#1}{\urlprefix }}%
\providecommand \urlprefix  [0]{URL }%
\providecommand \Eprint [0]{\href }%
\providecommand \doibase [0]{http://dx.doi.org/}%
\providecommand \selectlanguage [0]{\@gobble}%
\providecommand \bibinfo  [0]{\@secondoftwo}%
\providecommand \bibfield  [0]{\@secondoftwo}%
\providecommand \translation [1]{[#1]}%
\providecommand \BibitemOpen [0]{}%
\providecommand \bibitemStop [0]{}%
\providecommand \bibitemNoStop [0]{.\EOS\space}%
\providecommand \EOS [0]{\spacefactor3000\relax}%
\providecommand \BibitemShut  [1]{\csname bibitem#1\endcsname}%
\let\auto@bib@innerbib\@empty
\bibitem [{\citenamefont {{C. E. Shannon}}(1948)}]{shannon1948}%
  \BibitemOpen
  \bibfield  {author} {\bibinfo {author} {\bibnamefont {{C. E. Shannon}}},\
  }\href@noop {} {\bibfield  {journal} {\bibinfo  {journal} {Bell Sys. Tech.
  J.}\ }\textbf {\bibinfo {volume} {27}},\ \bibinfo {pages} {379} (\bibinfo
  {year} {1948})}\BibitemShut {NoStop}%
\bibitem [{\citenamefont {Cover}\ and\ \citenamefont
  {Thomas}(2006)}]{covertom}%
  \BibitemOpen
  \bibfield  {author} {\bibinfo {author} {\bibfnamefont {T.~M.}\ \bibnamefont
  {Cover}}\ and\ \bibinfo {author} {\bibfnamefont {J.~A.}\ \bibnamefont
  {Thomas}},\ }\href@noop {} {\emph {\bibinfo {title} {Elements of
  {I}nformation {T}heory}}},\ \bibinfo {edition} {2nd}\ ed.\ (\bibinfo
  {publisher} {Wiley},\ \bibinfo {address} {Hoboken, NJ, USA},\ \bibinfo {year}
  {2006})\BibitemShut {NoStop}%
\bibitem [{\citenamefont {Wilde}()}]{wilde-book}%
  \BibitemOpen
  \bibfield  {author} {\bibinfo {author} {\bibfnamefont {M.~M.}\ \bibnamefont
  {Wilde}},\ }\href {http://arxiv.org/abs/1106.1445} {\ }\bibinfo {note}
  {\href{http://arxiv.org/abs/1106.144}{http://arxiv.org/abs/1106.144}}\BibitemShut
  {NoStop}%
\bibitem [{\citenamefont {{T. Dorlas}}\ and\ \citenamefont {{C.
  Morgan}}(2011)}]{no-strong-converse-2011}%
  \BibitemOpen
  \bibfield  {author} {\bibinfo {author} {\bibnamefont {{T. Dorlas}}}\ and\
  \bibinfo {author} {\bibnamefont {{C. Morgan}}},\ }\href {\doibase
  http://link.aps.org/doi/10.1103/PhysRevA.84.042318} {\bibfield  {journal}
  {\bibinfo  {journal} {Phys. Rev. A}\ }\textbf {\bibinfo {volume} {84}},\
  \bibinfo {pages} {042318} (\bibinfo {year} {2011})}\BibitemShut {NoStop}%
\bibitem [{\citenamefont {{M. Berta}}\ \emph {et~al.}()\citenamefont {{M.
  Berta}}, \citenamefont {{O. Fawzi}},\ and\ \citenamefont {{S.
  Wehner}}}]{crypto-sconv-2011}%
  \BibitemOpen
  \bibfield  {author} {\bibinfo {author} {\bibnamefont {{M. Berta}}}, \bibinfo
  {author} {\bibnamefont {{O. Fawzi}}}, \ and\ \bibinfo {author} {\bibnamefont
  {{S. Wehner}}},\ }\href {http://arxiv.org/abs/1111.2026} {\ }\bibinfo {note}
  {\href{http://arxiv.org/abs/1111.2026}{http://arxiv.org/abs/1111.2026}}\BibitemShut
  {NoStop}%
\bibitem [{\citenamefont {{R. G. Gallager}}(1965)}]{gallager-expo-1965}%
  \BibitemOpen
  \bibfield  {author} {\bibinfo {author} {\bibnamefont {{R. G. Gallager}}},\
  }\href {http://dx.doi.org/10.1109/TIT.1965.1053730} {\bibfield  {journal}
  {\bibinfo  {journal} {IEEE Trans. Inf. Theory}\ }\textbf {\bibinfo {volume}
  {11}},\ \bibinfo {pages} {3} (\bibinfo {year} {1965})}\BibitemShut {NoStop}%
\bibitem [{\citenamefont {{S. Arimoto}}(1973)}]{arimoto-1973-converse}%
  \BibitemOpen
  \bibfield  {author} {\bibinfo {author} {\bibnamefont {{S. Arimoto}}},\ }\href
  {\doibase 10.1109/TIT.1973.1055007} {\bibfield  {journal} {\bibinfo
  {journal} {IEEE Trans. Inf. Theory}\ }\textbf {\bibinfo {volume} {19}},\
  \bibinfo {pages} {357 } (\bibinfo {year} {1973})}\BibitemShut {NoStop}%
\bibitem [{\citenamefont {Ogawa}\ and\ \citenamefont
  {Nagaoka}(1999)}]{ogawa-1999-converse}%
  \BibitemOpen
  \bibfield  {author} {\bibinfo {author} {\bibfnamefont {T.}~\bibnamefont
  {Ogawa}}\ and\ \bibinfo {author} {\bibfnamefont {H.}~\bibnamefont
  {Nagaoka}},\ }\href {\doibase 10.1109/18.796386} {\bibfield  {journal}
  {\bibinfo  {journal} {IEEE Trans. Inf. Theory}\ }\textbf {\bibinfo {volume}
  {45}},\ \bibinfo {pages} {2486} (\bibinfo {year} {1999})}\BibitemShut
  {NoStop}%
\bibitem [{\citenamefont {{A. Winter}}(1999)}]{winter-99-converse}%
  \BibitemOpen
  \bibfield  {author} {\bibinfo {author} {\bibnamefont {{A. Winter}}},\ }\href
  {\doibase 10.1109/18.796385} {\bibfield  {journal} {\bibinfo  {journal} {IEEE
  Trans. Inf. Theory}\ }\textbf {\bibinfo {volume} {45}},\ \bibinfo {pages}
  {2481} (\bibinfo {year} {1999})}\BibitemShut {NoStop}%
\bibitem [{\citenamefont {{R. K\"{o}nig}}\ and\ \citenamefont {{S.
  Wehner}}(2009)}]{konig-2009-converse}%
  \BibitemOpen
  \bibfield  {author} {\bibinfo {author} {\bibnamefont {{R. K\"{o}nig}}}\ and\
  \bibinfo {author} {\bibnamefont {{S. Wehner}}},\ }\href {\doibase
  10.1103/PhysRevLett.103.070504} {\bibfield  {journal} {\bibinfo  {journal}
  {Phys. Rev. Lett.}\ }\textbf {\bibinfo {volume} {103}},\ \bibinfo {pages}
  {070504} (\bibinfo {year} {2009})}\BibitemShut {NoStop}%
\bibitem [{\citenamefont {{W. Matthews}}\ and\ \citenamefont {{S.
  Wehner}}()}]{matt-wehner}%
  \BibitemOpen
  \bibfield  {author} {\bibinfo {author} {\bibnamefont {{W. Matthews}}}\ and\
  \bibinfo {author} {\bibnamefont {{S. Wehner}}},\ }\href
  {http://arxiv.org/abs/1210.4722} {\ }\bibinfo {note}
  {\href{http://arxiv.org/abs/1210.4722}{http://arxiv.org/abs/1210.4722}}\BibitemShut
  {NoStop}%
\bibitem [{\citenamefont {{R. E. Blahut}}(1976)}]{blahut1976}%
  \BibitemOpen
  \bibfield  {author} {\bibinfo {author} {\bibnamefont {{R. E. Blahut}}},\
  }\href {\doibase 10.1109/TIT.1976.1055576} {\bibfield  {journal} {\bibinfo
  {journal} {IEEE Trans. Inf. Theory}\ }\textbf {\bibinfo {volume} {22}},\
  \bibinfo {pages} {410} (\bibinfo {year} {1976})}\BibitemShut {NoStop}%
\bibitem [{\citenamefont {{T. S. Han}}\ and\ \citenamefont {{S.
  Verd\'{u}}}(1994)}]{han-verdu1994}%
  \BibitemOpen
  \bibfield  {author} {\bibinfo {author} {\bibnamefont {{T. S. Han}}}\ and\
  \bibinfo {author} {\bibnamefont {{S. Verd\'{u}}}},\ }\href {\doibase
  10.1109/18.335943} {\bibfield  {journal} {\bibinfo  {journal} {IEEE Trans.
  Inf. Theory}\ }\textbf {\bibinfo {volume} {40}},\ \bibinfo {pages} {1247}
  (\bibinfo {year} {1994})}\BibitemShut {NoStop}%
\bibitem [{\citenamefont {{N. Sharma}}(2008)}]{sharmapra2008}%
  \BibitemOpen
  \bibfield  {author} {\bibinfo {author} {\bibnamefont {{N. Sharma}}},\ }\href
  {\doibase 10.1103/PhysRevA.78.012322} {\bibfield  {journal} {\bibinfo
  {journal} {Phys. Rev. A}\ }\textbf {\bibinfo {volume} {78}},\ \bibinfo {eid}
  {012322} (\bibinfo {year} {2008})}\BibitemShut {NoStop}%
\bibitem [{\citenamefont {{Y. Polyanskiy}}\ and\ \citenamefont {{S.
  Verd\'{u}}}(2010)}]{polyanskiy-2010-converse}%
  \BibitemOpen
  \bibfield  {author} {\bibinfo {author} {\bibnamefont {{Y. Polyanskiy}}}\ and\
  \bibinfo {author} {\bibnamefont {{S. Verd\'{u}}}},\ }in\ \href@noop {} {\emph
  {\bibinfo {booktitle} {Proc. 48th Allerton Conf. Comm. Cont. Comp.}}}\
  (\bibinfo {address} {Monticello, USA},\ \bibinfo {year} {2010})\BibitemShut
  {NoStop}%
\bibitem [{\citenamefont {{F. Hiai}}\ \emph {et~al.}(2011)\citenamefont {{F.
  Hiai}}, \citenamefont {{M. Mosonyi}}, \citenamefont {{D. Petz}},\ and\
  \citenamefont {{C. B\'{e}ny}}}]{hiai-2011}%
  \BibitemOpen
  \bibfield  {author} {\bibinfo {author} {\bibnamefont {{F. Hiai}}}, \bibinfo
  {author} {\bibnamefont {{M. Mosonyi}}}, \bibinfo {author} {\bibnamefont {{D.
  Petz}}}, \ and\ \bibinfo {author} {\bibnamefont {{C. B\'{e}ny}}},\ }\href
  {\doibase 10.1142/S0129055X11004412} {\bibfield  {journal} {\bibinfo
  {journal} {Rev. Math. Phys.}\ }\textbf {\bibinfo {volume} {23}},\ \bibinfo
  {pages} {691} (\bibinfo {year} {2011})}\BibitemShut {NoStop}%
\bibitem [{\citenamefont {{I. Csisz\'{a}r}}(1995)}]{csiszar-1995-rate}%
  \BibitemOpen
  \bibfield  {author} {\bibinfo {author} {\bibnamefont {{I. Csisz\'{a}r}}},\
  }\href {\doibase 10.1109/18.370121} {\bibfield  {journal} {\bibinfo
  {journal} {IEEE Trans. Inf. Theory}\ }\textbf {\bibinfo {volume} {41}},\
  \bibinfo {pages} {26} (\bibinfo {year} {1995})}\BibitemShut {NoStop}%
\bibitem [{\citenamefont {{B. Schumacher}}(1996)}]{schumacher1996}%
  \BibitemOpen
  \bibfield  {author} {\bibinfo {author} {\bibnamefont {{B. Schumacher}}},\
  }\href {\doibase 10.1103/PhysRevA.54.2614} {\bibfield  {journal} {\bibinfo
  {journal} {Phys. Rev. A}\ }\textbf {\bibinfo {volume} {54}},\ \bibinfo
  {pages} {2614} (\bibinfo {year} {1996})}\BibitemShut {NoStop}%
\bibitem [{\citenamefont {{B. Schumacher}}\ and\ \citenamefont {{M. A.
  Nielsen}}(1996)}]{schumacher-nielsen-1996}%
  \BibitemOpen
  \bibfield  {author} {\bibinfo {author} {\bibnamefont {{B. Schumacher}}}\ and\
  \bibinfo {author} {\bibnamefont {{M. A. Nielsen}}},\ }\href {\doibase
  10.1103/PhysRevA.54.2629} {\bibfield  {journal} {\bibinfo  {journal} {Phys.
  Rev. A}\ }\textbf {\bibinfo {volume} {54}},\ \bibinfo {pages} {2629}
  (\bibinfo {year} {1996})}\BibitemShut {NoStop}%
\bibitem [{\citenamefont {{H. Barnum}}\ \emph {et~al.}(1998)\citenamefont {{H.
  Barnum}}, \citenamefont {{M. A. Nielsen}},\ and\ \citenamefont {{B.
  Schumacher}}}]{barnum-nielsen-1998}%
  \BibitemOpen
  \bibfield  {author} {\bibinfo {author} {\bibnamefont {{H. Barnum}}}, \bibinfo
  {author} {\bibnamefont {{M. A. Nielsen}}}, \ and\ \bibinfo {author}
  {\bibnamefont {{B. Schumacher}}},\ }\href {\doibase 10.1103/PhysRevA.57.4153}
  {\bibfield  {journal} {\bibinfo  {journal} {Phys. Rev. A}\ }\textbf {\bibinfo
  {volume} {57}},\ \bibinfo {pages} {4153} (\bibinfo {year}
  {1998})}\BibitemShut {NoStop}%
\bibitem [{\citenamefont {{H. Barnum}}\ \emph {et~al.}(2000)\citenamefont {{H.
  Barnum}}, \citenamefont {{E. Knill}},\ and\ \citenamefont {{M. A.
  Nielsen}}}]{barnum-knill-2000}%
  \BibitemOpen
  \bibfield  {author} {\bibinfo {author} {\bibnamefont {{H. Barnum}}}, \bibinfo
  {author} {\bibnamefont {{E. Knill}}}, \ and\ \bibinfo {author} {\bibnamefont
  {{M. A. Nielsen}}},\ }\href {\doibase 10.1109/18.850671} {\bibfield
  {journal} {\bibinfo  {journal} {IEEE Trans. Inf. Theory}\ }\textbf {\bibinfo
  {volume} {46}},\ \bibinfo {pages} {1317} (\bibinfo {year}
  {2000})}\BibitemShut {NoStop}%
\bibitem [{\citenamefont {{S. Lloyd}}(1997)}]{lloyd-1997}%
  \BibitemOpen
  \bibfield  {author} {\bibinfo {author} {\bibnamefont {{S. Lloyd}}},\ }\href
  {\doibase 10.1103/PhysRevA.55.1613} {\bibfield  {journal} {\bibinfo
  {journal} {Phys. Rev. A}\ }\textbf {\bibinfo {volume} {55}},\ \bibinfo
  {pages} {1613} (\bibinfo {year} {1997})}\BibitemShut {NoStop}%
\bibitem [{\citenamefont {{P. W. Shor}}(2002)}]{shor-cap-2002}%
  \BibitemOpen
  \bibfield  {author} {\bibinfo {author} {\bibnamefont {{P. W. Shor}}},\ }in\
  \href@noop {} {\emph {\bibinfo {booktitle} {MSRI Workshop on Quantum
  Computation}}}\ (\bibinfo {address} {Berkeley, CA, USA},\ \bibinfo {year}
  {2002})\BibitemShut {NoStop}%
\bibitem [{\citenamefont {{I. Devetak}}(2005)}]{devetak-2005}%
  \BibitemOpen
  \bibfield  {author} {\bibinfo {author} {\bibnamefont {{I. Devetak}}},\ }\href
  {\doibase 10.1109/TIT.2004.839515} {\bibfield  {journal} {\bibinfo  {journal}
  {IEEE Trans. Inf. Theory}\ }\textbf {\bibinfo {volume} {51}},\ \bibinfo
  {pages} {44} (\bibinfo {year} {2005})}\BibitemShut {NoStop}%
\bibitem [{\citenamefont {{P. Hayden}}\ \emph {et~al.}(2008)\citenamefont {{P.
  Hayden}}, \citenamefont {{P. W. Shor}},\ and\ \citenamefont {{A.
  Winter}}}]{hayden-shor-2008}%
  \BibitemOpen
  \bibfield  {author} {\bibinfo {author} {\bibnamefont {{P. Hayden}}}, \bibinfo
  {author} {\bibnamefont {{P. W. Shor}}}, \ and\ \bibinfo {author}
  {\bibnamefont {{A. Winter}}},\ }\href {\doibase 10.1142/S1230161208000079}
  {\bibfield  {journal} {\bibinfo  {journal} {Open Syst. Inf. Dyn.}\ }\textbf
  {\bibinfo {volume} {15}},\ \bibinfo {pages} {71} (\bibinfo {year}
  {2008})}\BibitemShut {NoStop}%
\bibitem [{\citenamefont {{I. Csisz\'{a}r}}(1972)}]{csiszar-f-1970}%
  \BibitemOpen
  \bibfield  {author} {\bibinfo {author} {\bibnamefont {{I. Csisz\'{a}r}}},\
  }\href {\doibase 10.1007/FBF02018661} {\bibfield  {journal} {\bibinfo
  {journal} {Period. Math. Hung.}\ }\textbf {\bibinfo {volume} {2}},\ \bibinfo
  {pages} {191} (\bibinfo {year} {1972})}\BibitemShut {NoStop}%
\bibitem [{\citenamefont {Hull}(2003)}]{hull-finance-book}%
  \BibitemOpen
  \bibfield  {author} {\bibinfo {author} {\bibfnamefont {J.}~\bibnamefont
  {Hull}},\ }\href@noop {} {\emph {\bibinfo {title} {Options, Futures, and
  Other Derivatives}}},\ \bibinfo {edition} {5th}\ ed.\ (\bibinfo  {publisher}
  {Prentice Hall},\ \bibinfo {address} {Upper SaddleRiver, NJ, USA},\ \bibinfo
  {year} {2003})\BibitemShut {NoStop}%
\bibitem [{\citenamefont {{M. Grassl}}\ \emph {et~al.}(1997)\citenamefont {{M.
  Grassl}}, \citenamefont {{Th. Beth}},\ and\ \citenamefont {{T.
  Pellizzari}}}]{grassl-1997}%
  \BibitemOpen
  \bibfield  {author} {\bibinfo {author} {\bibnamefont {{M. Grassl}}}, \bibinfo
  {author} {\bibnamefont {{Th. Beth}}}, \ and\ \bibinfo {author} {\bibnamefont
  {{T. Pellizzari}}},\ }\href {\doibase 10.1103/PhysRevA.56.33} {\bibfield
  {journal} {\bibinfo  {journal} {Phys. Rev. A}\ }\textbf {\bibinfo {volume}
  {56}},\ \bibinfo {pages} {33} (\bibinfo {year} {1997})}\BibitemShut {NoStop}%
\bibitem [{\citenamefont {{C. H. Bennett}}\ \emph {et~al.}(1996)\citenamefont
  {{C. H. Bennett}}, \citenamefont {{D. P. DiVincenzo}}, \citenamefont {{J. A.
  Smolin}},\ and\ \citenamefont {{W. K. Wootters}}}]{bennett-1996-pra}%
  \BibitemOpen
  \bibfield  {author} {\bibinfo {author} {\bibnamefont {{C. H. Bennett}}},
  \bibinfo {author} {\bibnamefont {{D. P. DiVincenzo}}}, \bibinfo {author}
  {\bibnamefont {{J. A. Smolin}}}, \ and\ \bibinfo {author} {\bibnamefont {{W.
  K. Wootters}}},\ }\href {\doibase 10.1103/PhysRevA.54.3824} {\bibfield
  {journal} {\bibinfo  {journal} {Phys. Rev. A}\ }\textbf {\bibinfo {volume}
  {54}},\ \bibinfo {pages} {3824} (\bibinfo {year} {1996})}\BibitemShut
  {NoStop}%
\bibitem [{\citenamefont {{C. H. Bennett}}\ \emph {et~al.}(1997)\citenamefont
  {{C. H. Bennett}}, \citenamefont {{D. P. DiVincenzo}},\ and\ \citenamefont
  {{J. A. Smolin}}}]{bennett-1997}%
  \BibitemOpen
  \bibfield  {author} {\bibinfo {author} {\bibnamefont {{C. H. Bennett}}},
  \bibinfo {author} {\bibnamefont {{D. P. DiVincenzo}}}, \ and\ \bibinfo
  {author} {\bibnamefont {{J. A. Smolin}}},\ }\href {\doibase
  10.1103/PhysRevLett.78.3217} {\bibfield  {journal} {\bibinfo  {journal}
  {Phys. Rev. Lett.}\ }\textbf {\bibinfo {volume} {78}},\ \bibinfo {pages}
  {3217} (\bibinfo {year} {1997})}\BibitemShut {NoStop}%
\bibitem [{\citenamefont {Sibson}(1969)}]{sibson-1969}%
  \BibitemOpen
  \bibfield  {author} {\bibinfo {author} {\bibfnamefont {R.}~\bibnamefont
  {Sibson}},\ }\href {http://dx.doi.org/10.1007/BF00537520} {\bibfield
  {journal} {\bibinfo  {journal} {Prob. Theory Rel. Fields}\ }\textbf {\bibinfo
  {volume} {14}},\ \bibinfo {pages} {149} (\bibinfo {year} {1969})}\BibitemShut
  {NoStop}%
\end{thebibliography}%
\appendix
\section{Appendix}
\section{Proof of Lemma $1$}

Note that for any $\delta > 0$, there exists a $\sigma^B$ such that
$K_{\lambda}(A \rangle B)_\rho \geq D_{\lambda}(\rho^{AB} || \mathbbm{1} \otimes \sigma^B) - \delta$.
Using the monotonicity property from (2) in the main text, we have
$K_{\lambda}(A \rangle B)_\rho
\geq D_{\lambda} \left[ \rho^{AC} || \mathbbm{1} \otimes \cE^{B \to C}(\sigma^B) \right] - \delta
\geq \inf_{\sigma^C} D_{\lambda}(\rho^{AC} || \mathbbm{1} \otimes \sigma^C) - \delta
= K_{\lambda}^{(c)}(A \rangle C)_\rho - \delta$.
Since this is true for any $\delta > 0$, the result follows.

\section{Proof of Lemma $2$}

The proof of Lemma $2$ follows straightforwardly from the definition of
${\cK}_{\lambda}(A \rangle B)$
and from the following Lemma.
\begin{lemma}[Quantum Sibson identity]
\label{qsibson}
For any quantum state $\rho^{AB}$ in system $AB$ and $D_\lambda$ as
the \renyi divergence of order $\lambda$, we have
\begin{align*}
D&_{\lambda} (\rho^{AB}||\mathbbm{1}\otimes \sigma^{B}) \nonumber\\ 
& = D_{\lambda}(\sigma^* ||\sigma^{B}) + \frac{\lambda}{\lambda-1}\log \tr \Big[\tr_{A}
\Big(\rho^{AB}\Big)^{\lambda}\Big]^{\frac{1}{\lambda}},
\end{align*}
\begin{equation*}
\mbox{where } ~~~~ \sigma^{*} =
\frac{\left[\tr_{A}\left(\rho^{AB}\right)^\lambda\right]^{\frac{1}{\lambda}}}{\tr\left[\tr_{A}\left(\rho^{AB}\right)^\lambda\right]^{\frac{1}{\lambda}}}.
\end{equation*}
\end{lemma}
\begin{proof}
For the classical Sibson identity, see Ref. \cite{sibson-1969}. Note that
\begin{align*}
D_{\lambda} &(\rho^{AB}||\mathbbm{1}\otimes \sigma^{B})\\ 
& = \frac{1}{\lambda-1} \log \tr \left(\rho^{AB}\right)^{\lambda} [ \mathbbm{1}\otimes (\sigma^{B})^{1-\lambda} ] \\
& = \frac{1}{\lambda-1}\log\tr \,
\tr_A\left(\rho^{AB}\right)^{\lambda}(\sigma^{B})^{1-\lambda} \\
& = \frac{1}{\lambda-1}\log\tr \left(\sigma^{*}\right)^{\lambda}(\sigma^{B})^{1-\lambda}\\
& \hspace{5mm}+ \frac{\lambda}{\lambda-1}\log \tr \left[\tr_{A}\left(\rho^{AB}\right)^{\lambda}\right]^{{\frac{1}{\lambda}}} \\
& =  D_{\lambda}(\sigma^{*}||\sigma_{B}) + \frac{\lambda}{\lambda-1}\log\tr \left[\tr_{A}\left(\rho^{AB}\right)^{\lambda}\right]^{{\frac{1}{\lambda}}}.
\end{align*}
Since $D_{\lambda}(\sigma^{*}||\sigma_{B}) \geq 0$, choosing $\sigma_B = \sigma^*$ gives
us the minimum and the result follows.
\end{proof}

\section{Proof of Theorem $2$}
For any quantum state $\sigma^{AB}$, $s \in [-1/2,0)$, let
\begin{align*}
g(s) := -\log \tr \left[ \tr_A \left( \sigma^{AB} \right)^{1/(1+s)} \right]^{s+1}
\end{align*}
It easily follows that $g(s) = 0$. To show that
${\partial g(s)}/{\partial s}\big|_{s=0} = I(A\rangle B)_{\sigma}$, we use
the following differentiation rule (Lemma 4 in Ref. \cite{ogawa-1999-converse})
for a Hermitian operator $X(s)$ parametrized by a real parameter $s$
\[
\frac{\partial}{\partial s}\tr g[X(s)] = \tr g^\prime[X(s)] \frac{\partial X(s)}{\partial s}.
\]
Let the spectral decomposition of $\sigma^{AB}$ be
$\sigma^{AB} = \sum_{i}\lambda_i\ket{i} \bra{i}^{AB}$
and let $\sigma_i = \tr_{A} \ket{i} \bra{i}^{AB}$.
Hence, we get
$\sigma^{B} = \tr_{A}\sigma^{AB} = \sum_{i}\lambda_{i}\sigma_i$ and
$\kappa_{1} := \tr_A(\sigma^{AB})^{1/(s+1)} = \sum_{i}\lambda_{i}^{1/(s+1)}
\sigma_i$.
It is easy to see that
${\partial \kappa_1}/{\partial s} = -{\kappa_2}/{(s+1)}$,
where $\kappa_2 = \sum_{i}\lambda_i ^{\frac{1}{s+1}}\log (\lambda_i ^{\frac{1}{s+1}})
\sigma_i$. It now follows that 
\begin{align*}
\frac{\partial g(s)}{\partial s} & =
\frac{\tr \kappa_{1}^{s}(\kappa_2-\kappa_1\log \kappa_1)}{\tr\kappa_1^{s+1}}, \\
\frac{\partial g(s)}{\partial s} \Big |_{s=0} &= \tr \Big[ \sum_{i}\lambda_{i} (\log \lambda_i)
\sigma_i- \big( \sum_i \lambda_i\sigma_i \big)\\
&\hspace{10mm} \log \big(\sum_{i}\lambda_i\sigma_i \big)\Big],\\
& = H(B)_{\sigma}-H(A,B)_{\sigma},\\
& = I(A\rangle B)_{\sigma}.
\end{align*}
We now show that $g(s) + (s+1) \log |A|$ is an increasing function in $s$. Consider
the operators $E_i = \sqrt{\sigma_i/|A|}$. Then $\sum_i E_i^\dagger E_i =$
$\sum_i \tr_A \ket{i} \bra{i}^{AB}/|A|$ $= {\mathbbm{1}}$. Since $x^\gamma$,
$\gamma \in (0,1]$ is operator concave, we have, using the operator Jensen's
inequality and for $1/2 \leq \alpha \leq \beta < 1$, $\gamma = \alpha/\beta$,
\beq
\left( \frac{1}{|A|} \sum_i \lambda_i^{1/\beta} \sigma_i \right)^\beta
\leq \left( \frac{1}{|A|} \sum_i \lambda_i^{1/\alpha} \sigma_i \right)^\alpha,\nonumber
\enq
or $g(\alpha-1) + \alpha \log |A| \leq g(\beta-1) + \beta \log |A|$.

\section{An Alternate Proof for Theorem $3$ using the hockey-stick divergence}

Note that the following set of inequalities hold for the hockey stick divergence.
\begin{fleqn}
\begin{align*}
\cD(& \rho^{A B^n} || \mathbbm{1} \otimes \rho^{B^n} ) \\
 &\stackrel{a}{=} \sum_{2^n {\mathrm{terms}}}
\alpha_{k,n} \cD ( \zeta_{i_1,...,i_{n-k}}^{AB_{i_1} \cdots B_{i_n}} ||
\mathbbm{1} \otimes \zeta_{i_1,...,i_{n-k}}^{B_{i_1} \cdots B_{i_n}} ) \\
&  \stackrel{b}{=} \sum_{2^n {\mathrm{terms}}}
\alpha_{k,n} \cD ( \sigma^{AB_{i_1} \cdots B_{i_{n-k}}} ||
\mathbbm{1} \otimes
\sigma^{B_{i_1} \cdots B_{i_{n-k}}} ) \\
&  \stackrel{c}{\leq} \sum_{2^n {\mathrm{terms}}}
\alpha_{k,n} \cD ( \rho^{AA^\prime_{i_1} \cdots A^\prime_{i_{n-k}}} ||
\mathbbm{1} \otimes \rho^{A^\prime_{i_1} \cdots A^\prime_{i_{n-k}}} ),
\end{align*}
\end{fleqn}
where $a$ follows from orthogonality of $\zeta$'s, $b$ follows since we have
removed the tensors with $\ket{e}\bra{e}$, and $c$ follows from monotonicity
(see Lemma \ref{mono-hstick}).
Using the above, we now have
\begin{fleqn}
\begin{align*}
\cK(A \rangle B^n) & \leq  \cD(\rho^{A B^n} || \mathbbm{1} \otimes \rho^{B^n} ) \\
& \leq \sum_{k=0}^n
{n \choose k} \alpha_{k,n}  \\
&\hspace{6mm} \tr \Big( \rho^{AA^\prime_1 \cdots A^\prime_{n-k}} -
\gamma \mathbbm{1} \otimes
\rho^{A^\prime_1 \cdots A^\prime_{n-k}} \Big)^+ \\
& \leq \sum_{k=0}^{\frac{n}{2} - \lfloor \frac{\log \gamma}{2\log d_A} \rfloor }
{n \choose k} \alpha_{k,n},
\end{align*}
\end{fleqn}
where we have upper bounded $\tr ( \rho^{AA^\prime_1 \cdots A^\prime_{n-k}} -
\gamma \mathbbm{1} \otimes \rho^{A^\prime_1 \cdots A^\prime_{n-k}} )^+$ by $1$
for $k \leq n/2 - \lfloor \log \gamma/(2\log d_A) \rfloor$.
Choose $\log \gamma = n[\cR + Q(\cN)]/2$ in the above equation.
For $\cR > Q(\cN)$, we have $n/2 - \lfloor \log \gamma/(2\log d_A) \rfloor < np$.
Similar to the quantity defined in Property 2 in the main text, we define ${\mathbbmss{D}}$.
Let $\Pi_0 = \ket{0}\bra{0}$ and $\Pi_1 = \ket{1}\bra{1}$
be two projectors with $\Pi_0 + \Pi_1 = \mathbbm{1}$.
Let $\alpha \in [0,1]$, $\beta \in (0,1]$, $\rho = \alpha \Pi_0 + (1-\alpha) \Pi_1$,
$\sigma = \beta \Pi_0 + (1/\beta-\beta) \Pi_1$, and let us define
\beq
{\mathbbmss{D}}(\alpha || \beta) := \mathcal{D}(\rho || \sigma).
\enq
Using the Chernoff bound, the inequality
${\mathbbmss{D}}(\mathbb{F} || e^{-n\cR})  \geq \mathbb{F} - \gamma e^{-n\cR}$,
and Theorem 1 in the main text, we get
\begin{align*}
\mathbb{F} \leq &\exp \left\{-\frac{n}{2}[\cR-Q(\cN)] \right\}\nonumber\\ 
&\hspace{1mm}+\exp \Big\{-\frac{n}{2p} \Big[ \frac{(2p-1)^+}{2} 
+\frac{\cR}{4 \log d_A} \Big]^2 \Big\},\nonumber
\end{align*}
which gives us the strong converse. \\

\section{Monotonicity lemma}

\begin{lemma}
\label{mono-hstick}
Consider the matrices $\rho, \sigma \geq 0$ and a scalar $\gamma > 0$. Then for any CPTP map
$\cE$,
\beq
\tr(\rho - \gamma \sigma)^+ \geq \tr \left[ \cE(\rho) - \gamma \cE(\sigma) \right]^+. \nonumber
\enq
\end{lemma}
\begin{proof}
Let the Jordan decomposition of $\rho - \gamma \sigma = Q - S$, where $Q,S \geq 0$. Let
$P := P_{\{ \cE(\rho) - \gamma \cE(\sigma) \geq 0 \}}$. Then
\begin{align*}
\tr(\rho - \gamma \sigma)^+ & = \tr Q \\
& \stackrel{a}{=} \tr \cE(Q)  \\
& \stackrel{b}{\geq} \tr P [\cE(Q) - \cE(S)]  \\
& = \tr \left[ \cE(\rho) - \gamma \cE(\sigma) \right]^+,
\end{align*}
where $a$ follows since $\cE$ is trace preserving, $b$ follows since
we are subtracting non-negative terms.
\end{proof}

\end{document}